\documentclass[letterpaper, 10 pt, conference]{ieeeconf}  % Comment this line out if you need a4paper

\IEEEoverridecommandlockouts                              % This command is only needed if 
\usepackage{graphics} % for pdf, bitmapped graphics files

\usepackage{algorithm}
\usepackage{algpseudocode}
\usepackage{mathtools}
\usepackage{mathbbol}
\usepackage[hidelinks, raiselinks=true, bookmarksnumbered=true, pdfusetitle, hyperfootnotes=false, backref=page]{hyperref} % add hyperlinks 
\usepackage[capitalise]{cleveref} % use "\cref{...}" instead of "Bla.~\ref{...}"
\usepackage{ifthen}
\usepackage{siunitx} % for \num command
\usepackage{booktabs}
\usepackage{xcolor} % for \textcolor
\newboolean{extendedversion}
\setboolean{extendedversion}{true}

\newtheorem{prop}{Proposition}
\newtheorem{theo}{Theorem}
\newtheorem{ass}{Assumption}

\newtheorem{lem}{Lemma}

\title{\LARGE \bf
Statistically Consistent Approximate Model Predictive Control
}

\author{E. Milios$^{1,2}$, K. P. Wabersich$^{1}$, F. Berkel$^{1}$, F. Gruber$^{1}$, and M. N. Zeilinger$^{2}$ % <-this % stops a space
\thanks{$^{1}$ Elias Milios, Kim P. Wabersich, Felix Berkel, and Felix Gruber are with the Corporate Research of Robert Bosch GmbH, Renningen, Germany.
        {\tt\small elias.milios, kimpeter.wabersich, felix.berkel, felix.gruber@de.bosch.com}}%
\thanks{$^{2}$Elias Milios and Melanie N. Zeilinger are with the Institute of Dynamic Systems and Control, ETH Zurich, Zurich, Switzerland.
        {\tt\small emilios,mzeilinger@ethz.ch}}%
}

\begin{document}

\maketitle
\thispagestyle{empty}
\pagestyle{empty}

%%%%%%%%%%%%%%%%%%%%%%%%%%%%%%%%%%%%%%%%%%%%%%%%%%%%%%%%%%%%%%%%%%%%%%%%%%%%%%%%
\begin{abstract}
    Model Predictive Control (MPC) offers rigorous safety and performance guarantees but is computationally intensive. Approximate MPC (AMPC) aims to circumvent this drawback by learning a computationally cheaper surrogate policy. Common approaches focus on imitation learning (IL) via behavioral cloning (BC)\ifthenelse{\boolean{extendedversion}}{, minimizing a mean-squared-error loss on a collection of state-input pairs}. However, BC fundamentally fails to provide accurate approximations when MPC solutions are set-valued due to non-convex constraints or local minima. We propose a two-stage IL procedure to accurately approximate nonlinear, potentially set-valued MPC policies. The method integrates an approximation of the MPC's optimal value function into a one-step look-ahead loss function, and thereby embeds the MPC's constraint and performance objectives into the IL objective. \ifthenelse{\boolean{extendedversion}}{This is achieved by adopting a stabilizing soft constrained MPC formulation, which reflects constraint violations in the optimal value function by combining a constraint tightening with slack penalties.}{} We prove statistical consistency for policies that exactly minimize our IL objective, implying convergence to a safe and stabilizing control law, and establish input-to-state stability guarantees for approximate minimizers. Simulations demonstrate improved performance compared to BC.
\end{abstract}

%\begin{IEEEkeywords}
%
%\end{IEEEkeywords}

%%%%%%%%%%%%%%%%%%%%%%%%%%%%%%%%%%%%%%%%%%%%%%%%%%%%%%%%%%%%%%%%%%%%%%%%%%%%%%%%
\section{Introduction}
MPC is a powerful control framework for modern control systems, as it provides nearly optimal policies which ensure closed-loop stability while satisfying state and input constraints. Yet, its widespread industrial deployment is hindered by the need for online optimization, which demands substantial computational resources and significantly complicates controller certification for real-time deployment. To address this, AMPC seeks to approximate the MPC policy with a surrogate that is tailored for the use in embedded devices. To this end, a popular approach is to learn a neural network (NN) by minimizing a mean-squared-error (MSE) loss on a collection of sampled state-input pairs, referred to as BC. While proving to be effective in many settings~\cite{AMPC_survey_Gonzales}, BC-based approaches show systematic errors when presented with set-valued MPC solutions~\cite{Bishop_mixture_density_networks}, which can arise from over-actuated systems, non-convexities, or local minima in nonlinear MPC~\cite{Diffusion_AMPC}. Resulting approximation errors compromise the original safety and stability guarantees, potentially leading to unsafe and poorly performing closed-loop behaviors.\par
To this end, several works focus on recovering system-theoretic guarantees despite approximation errors. In~\cite{AMPC_Hertneck}, an input-robust MPC is utilized to account for potential approximation errors. However, the work remains isolated from the learning procedure and assumes a sufficiently small approximation error can be achieved, thereby neglecting the shortcomings of BC. Instead of approximating the MPC policy, the authors in~\cite{AMPC_Nicolas} circumvent the drawbacks of BC by approximating the unique optimal value function of a tailored soft constrained MPC (SCMPC) formulation. This approximate value function is used to formulate a one-step look-ahead objective, returning an approximate MPC policy through its online minimization. While this enables the recovery of stability and constraint satisfaction guarantees, the required online optimization can still be limiting, especially for high dimensional input spaces. Finally, works such as~\cite{Drgona_DiffPC} and~\cite{AMPC_Ghezzi} aim to avoid the drawbacks of BC through MPC-inspired loss functions. Specifically,~\cite{Drgona_DiffPC} takes a Lagrangian relaxation of the MPC problem as a loss function, preserving structural characteristics of the MPC to guide the learning of the approximate policy. However, relying on soft penalties involves careful tuning to balance performance, safety, and stability. To circumvent this, the work in~\cite{AMPC_Ghezzi} directly embeds the MPC problem as the IL loss function, evaluating the cost of a candidate control input by solving the original MPC problem with the candidate input used for the initial prediction step. This approach preserves the MPC's exact performance and constraint characteristics, however, it induces significant computational overhead due to repeated MPC solving during the training process.\par
%

%Finally, works such as~\cite{Diffusion_AMPC} and~\cite{Mixture_Gaussian_AMPC} specifically deal with approximating set-valued MPC policies. Conceptually, these works differ from the present work by aiming to recover the full set-valued mapping, rather than one element thereof. Additionally, while appearing promising when applied in real-world settings~\cite{Diffusion_AMPC}, the lack of system-theoretic guarantees may limit the applicability of these approaches in safety-critical environments.
%
\textit{Contribution:}
This work proposes an IL procedure designed to mimic nonlinear, potentially set-valued MPC policies. Our approach provides theoretical guarantees for both the learning process and the resulting policy, while accounting for practical limitations such as finite data and imperfect optimization. Specifically, our contributions are:\par
1) Statistical consistency: Any policy that minimizes the expected value of the proposed loss over a user-defined data set will converge to a safe and stabilizing optimal control law with arbitrary accuracy as the number of training samples increases.\par
2) Stability guarantees: Input-to-state stability (ISS)~\cite{ISS_Sontag} is guaranteed even for policies that only approximately minimize the proposed IL objective on finite data, reflecting realistic conditions such as early stopping and limited hyperparameter tuning.\par
The proposed IL procedure consists of two stages, conceptually merging established AMPC approaches. At the first stage, motivated by~\cite{AMPC_Nicolas}, we approximate the optimal value function of a stabilizing SCMPC formulation. At the second stage, building on the idea of employing MPC-inspired loss functions~\cite{Drgona_DiffPC, AMPC_Ghezzi}, we use the approximate value function to construct a differentiable and fast to evaluate loss function that inherits the original MPC objectives. This enables efficient training of an explicit approximate MPC policy and, consequently, extends the approaches in~\cite{AMPC_Nicolas} and~\cite{AMPC_Ghezzi} by reducing the computational complexity of online policy evaluation and offline policy training, respectively. Moreover, our statistical consistency guarantees remain applicable to the setting considered in~\cite{AMPC_Ghezzi}, thereby providing an additional, complementary extension to~\cite{AMPC_Ghezzi}.
\par
\ifthenelse{\boolean{extendedversion}}{\textit{Outline:}
\cref{sec: problem formulation} introduces the problem formulation. \cref{sec: approximation procedure} describes the proposed IL procedure and \cref{sec: theoretical analysis} provides its theoretical analysis. Specifically,~\cref{sec: statistical properties} derives statistical consistency of policies that exactly minimize the proposed IL objective, and~\cref{sec: stability} establishes ISS for policies that achieve only approximate minimization. Finally, \cref{sec: numerical experiments} concludes with the numerical experiments.}{}

\textit{Notation:} \ifthenelse{\boolean{extendedversion}}{For two vectors $a \in \mathbb{R}^n$ and $b \in \mathbb{R}^m$, the ordered pair $(a,b) \in \mathbb{R}^{n+m}$ represents their concatenation.}{} The symbol $\mathbb{1}$ denotes a vector of ones of appropriate dimension. For \ifthenelse{\boolean{extendedversion}}{}{a vector $a \in \mathbb{R}^n$ and} sets $\mathcal{A},\mathcal{B} \subset \mathbb{R}^n$ we denote by $\mathbb{d}(a,\mathcal{A}):=\inf_{a'\in \mathcal{A}}\lVert a - a'\rVert$ the distance from $a$ to $\mathcal{A}$, and by $\mathbb{D}(\mathcal{A},\mathcal{B}):=\sup_{a \in \mathcal{A}}\mathbb{d}(a,\mathcal{B})$ the deviation from $\mathcal{A}$ to $\mathcal{B}$. Given a function $f: \mathbb{R}^n \rightarrow \mathbb{R}^m$ and a collection of vectors $\mathcal{C}=\{a_i\}_{i=1}^{N_\mathcal{C}}$ with $a_i \in \mathbb{R}^n$ and $N_\mathcal{C}\geq1$, we define the empirical expectation of $f(a)$ over $\mathcal{C}$ by $\mathbb{E}_\mathcal{C}[f(a)]:=N_\mathcal{C}^{-1}\sum_{i=1}^{N_\mathcal{C}}f(a_i)$. Similarly, given a probability distribution $\mathcal{P}$ with probability density function $p(a)$ and support $\mathcal{S}\subseteq \mathbb{R}^n$, we denote the expected value of $f(a)$ when $a$ is sampled from $\mathcal{P}$ on the domain $\mathcal{S}$ by $\mathbb{E}_{\mathcal{S}}[f(a)] = \int_{\mathcal{S}}f(a)p(a)da$. A continuous function $\alpha: [0,s) \rightarrow \mathbb{R}_{\geq 0}$ is called a $\mathcal{K}$-function if it is strictly increasing and $\alpha(0) = 0$. A continuous function $\beta:[0,s) \times [0,\infty) \rightarrow [0,\infty)$ is called a $\mathcal{KL}$-function if for each fixed $s$, $\beta(r,s)$ is a $\mathcal{K}$-function and for each fixed $r$, $\beta(r,s)$ is decreasing with respect to $s$ and $\beta(r,s) \rightarrow 0$ as $s \rightarrow \infty$. 
\section{Preliminaries}\label{sec: problem formulation}
We consider nonlinear discrete-time systems of the form
\begin{equation}\label{eq:nonlinear_system}
    x(k+1) = f(x(k),u(k)),
\end{equation}
where $f:\mathbb{R}^{n_x} \times \mathbb{R}^{n_u} \rightarrow \mathbb{R}^{n_x}$ denotes the system dynamics, $x(k) \in \mathbb{R}^{n_x}$ and $u(k) \in \mathbb{R}^{n_u}$ the system state and control input, respectively, and $k \in \mathbb{N}$ the discrete time index. We assume that $f$ is differentiable and satisfies $f(0,0) = 0$. Furthermore, we consider polytopic state-input constraints of the form $x(k) \in \mathbb{X}:=\{x \in \mathbb{R}^{n_x} \mid H_xx\leq \mathbb{1}\}$, and $u(k) \in \mathbb{U}:=\{u \in \mathbb{R}^{n_u} \mid H_uu \leq \mathbb{1}\}$, where $H_x \in \mathbb{R}^{m_x \times n_x}$, $H_u \in \mathbb{R}^{m_u \times n_u}$.
\subsection{Soft Constrained MPC}
Given system~\eqref{eq:nonlinear_system}, we formulate an MPC problem that facilitates its approximation using IL. To this end, we adopt the following SCMPC problem~\cite{AMPC_Nicolas}:
\begin{subequations}\label{eq: generic MPC}
\begin{align}
    V_\mathrm{MPC}(x(k)) & :=  \min_{u_{\cdot|k}, \xi_{\cdot|k},\alpha} \ J_\mathrm{p}(x(k),u_{\cdot|k}) + J_\mathrm{\xi}(\xi_{\cdot|k})\label{eq: generic MPC cost}\\
    \text{s.t. } \  & x_{0|k} = x(k),\label{eq: initial state constraint}\\
    & x_{N|k} \in \alpha \mathbb{X}_\mathrm{f}, \quad 0 \leq \alpha \leq 1, \quad \xi_{N|k} \geq 0,\label{eq: term constraint} \\
    & \alpha \mathbb{X}_\mathrm{f} \subseteq \{x \mid H_xx \leq \mathbb{1}(1-\eta)+\xi_{N|k}\},\label{eq: scaled terminal set}\\
    & \forall i \in \{0,1,\ldots,N-1\}:\notag \\
    & \quad x_{i+1|k} = f(x_{i|k},u_{i|k}), \\
    & \quad u_{i|k} \in \mathbb{U}, \quad \xi_{i|k}\geq 0, \\
    & \quad H_xx_{i|k} \leq \mathbb{1}(1-\eta) + \xi_{i|k} + \xi_{N|k}. \label{eq: slacked state constraints}
\end{align}
\end{subequations}
Here, we optimize over a predicted input, state, and slack variable sequence $u_{\cdot|k} = \{u_{i|k}\}_{i=0}^{N-1}$, $x_{\cdot|k} = \{x_{i|k}\}_{i=0}^{N}$, and $\xi_{\cdot|k} = \{\xi_{i|k}\}_{i=0}^{N}$ with $\xi_{i|k} \in \mathbb{R}^{m_x}$ and $i$ denoting the prediction time step.
Akin to a common MPC formulation, in \eqref{eq: generic MPC cost} we consider the performance cost $J_\mathrm{p}(x_{\cdot|k},u_{\cdot|k}) := \sum_{i=0}^{N-1} \ell(x_{i|k},u_{i|k}) + V_\mathrm{f}(x_{N|k})$ with stage cost $\ell: \mathbb{R}^{n_x} \times \mathbb{R}^{n_u} \rightarrow \mathbb{R}_{\geq 0}$, terminal cost $V_\mathrm{f}:\mathbb{R}^{n_x} \rightarrow \mathbb{R}_{\geq 0}$, and prediction horizon $N$. In addition,~\eqref{eq: generic MPC cost} contains the penalty cost $J_\xi(\xi_{\cdot|k}):= \rho\ell_\xi(\xi_{N|k})+\rho\sum_{i=0}^{N-1}\ell_\xi(\xi_{N|k}+\xi_{i|k})$ with penalty function $\ell_\xi:\mathbb{R}^{m_x} \rightarrow\mathbb{R}_{\geq 0}$ and scaling factor $\rho > 0$, introducing large cost values in the case of state constraint violations. Ideally, as noted in~\cite[Sec. 3]{AMPC_Nicolas}, $\ell_\xi$ represents an exact penalty function to recover the corresponding hard-constrained MPC when feasible.
Equation \eqref{eq: term constraint} implements a terminal constraint with ellipsoidal terminal set $\mathbb{X}_{\mathrm{f}}:= \{x \in \mathbb{R}^{n_x} \mid x^\top P x \leq h_\mathrm{f}\}$, $h_\mathrm{f}>0$, $P \succ 0$, and terminal scaling factor $\alpha \in [0,1]$. In contrast to classical MPC, we do not require $\mathbb{X}_\mathrm{f}$ to be contained in $\mathbb{X}$. Instead, with~\eqref{eq: scaled terminal set} we require $\alpha\mathbb{X}_\mathrm{f}$ to be contained in the softened state constraint set determined by the terminal slack variable $\xi_{N|k}$. This plays a key role in establishing asymptotic stability~\cite{AMPC_Nicolas}. The constraint \eqref{eq: scaled terminal set} implements a constraint tightening factor $\eta \in (0,1]$, which is used to prevent actual constraint violations.\par
Formulation~\eqref{eq: generic MPC} offers important advantages over classical hard-constrained alternatives. The soft constraints enable data collection from regions beyond the hard state constraint set without compromising stability, thereby enhancing data efficiency, mitigating data sparsity near constraint boundaries, and explicitly reflecting constraint violations in the value function. Moreover, if $J_\mathrm{P}$ is Lipschitz continuous, ~\eqref{eq: generic MPC} ensures continuity of $V_\mathrm{MPC}$ over bounded sets where the constraint $\alpha \leq 1$ in~\eqref{eq: term constraint} is inactive~\cite[Theorem 1]{AMPC_Nicolas}. This ensures that $V_\mathrm{MPC}$ can be arbitrarily well approximated by common NN architectures and standard regression techniques, which becomes important in \cref{sec: approximation procedure}.\par
We define the set of feasible input sequences to problem~\eqref{eq: generic MPC} by $\mathcal{U}_{N}(x(k)):=\{u_{\cdot|k} \in \mathbb{U}^{n_u \times N} \mid \exists \xi_{\cdot|k}, \alpha \  \mathrm{ s.t. } \ \eqref{eq: initial state constraint}~-~\eqref{eq: slacked state constraints}\}$ and the set of feasible states by $\mathcal{X}_N:=\{x(k) \in \mathbb{R}^{n_x} \mid \mathcal{U}_N(x(k)) \neq \emptyset\}$. Then, for any $x \in \mathcal{X}_N$, we denote an optimal input sequence to~\eqref{eq: generic MPC} by $u_{\cdot|k}^\mathrm{MPC}=\{u^\mathrm{MPC}_{i|k}\}_{i=0}^{N-1}$. Based on this, the set of MPC inputs is defined by $\mathcal{U}_\mathrm{MPC}(x(k)):=\{u^\mathrm{MPC}_{0|k} \mid \exists \{u^\mathrm{MPC}_{i|k}\}_{i=0}^{N-1} \in \mathcal{U}_N(x(k))\}$.
Importantly, we allow~\eqref{eq: generic MPC} to admit multiple optimal inputs for a given state, i.e., we do not assume $\mathcal{U}_\mathrm{MPC}(x)$ to be a singleton. Consequently, the resulting MPC input $\pi_{\mathrm{MPC}}(x) \in \mathcal{U}_\mathrm{MPC}(x)$ can vary based on the solver's initialization and internal mechanics, potentially yielding different solutions for identical or neighboring states. We model this similar to~\cite{Diffusion_AMPC} as a random process, where the solver's selection $\pi_\mathrm{MPC}(x)$ is determined by the conditional probability density function $p_\mathrm{MPC}(u\mid x)$ whose support is given by $\mathcal{U}_\mathrm{MPC}(x)$.\par
We assume that the ingredients of~\eqref{eq: generic MPC} are chosen such that application of $\pi_\mathrm{MPC}(x) \sim p_\mathrm{MPC}(u \mid x)$ ensures stability.
\begin{ass}\label{ass: MPC stability}
    For all $x(k) \in \mathcal{X}_N$ and $k \geq0$, application of any $u(k) = \pi_\mathrm{MPC}(x(k)) \in \mathcal{U_\mathrm{MPC}}(x(k))$ implies asymptotic stability of the origin $(x,u) = (0,0)$ of system~\eqref{eq:nonlinear_system}.
\end{ass}\par
Sufficient conditions for~\cref{ass: MPC stability} are provided, e.g., in~\cite{AMPC_Nicolas}, which are closely aligned with common stabilizing MPC assumptions.\par
\subsection{Problem Formulation and Objectives}
The primary goal of this work is to train a policy $\pi(x, \theta_\pi)$ with parameters $\theta_\pi \in \Theta_\pi$ and compact parameter set $\Theta_\pi \subset \mathbb{R}^{n_p}$ to approximate the behavior of~\eqref{eq: generic MPC}.
To this end, the common approach is naive BC. Conceptually, the goal of BC is to find a policy $\pi(x,\theta_\pi^{\mathrm{BC}})$ satisfying $\pi(x,\theta_\pi^{\mathrm{BC}}) \approx \pi_\mathrm{MPC}(x)$ by trying to match \textit{each} specific MPC input sample $\pi_\mathrm{MPC}(x)$ collected in the training data. Formally, BC returns $\pi(x,\theta_\pi^{\mathrm{BC}})\approx \int_{\mathcal{U}_\mathrm{MPC}(x)} u \: p_\mathrm{MPC}(u \mid x)du$ and, hence, learns a policy that is systematically biased towards the conditional mean of the MPC solutions sampled around a state~\cite{Bishop_mixture_density_networks}. The mean, however, can differ substantially from each element in $\mathcal{U}_\mathrm{MPC}(x)$. Consequently, the learned BC policy can have large distance to the set of MPC solutions, i.e., $\mathbb{d}(\pi(x,\theta_\pi^{\mathrm{BC}}), \mathcal{U}_\mathrm{MPC}(x))$ can be large.
To overcome this drawback, we aim to obtain an approximate policy that is close to \textit{at least one} optimal MPC input. Formally, we aim to learn a policy that, for each $x \in \mathcal{X}_N$, satisfies $\mathbb{d}(\pi(x,\theta_\pi), \mathcal{U}_\mathrm{MPC}(x)) \approx 0$.
Specifically, our objectives are twofold:\par
%%%%%%%%%%%%%%%
1.) Given a data set  $\mathcal{D}_X:=\{x_j\}_{j=1}^{N_s}$ collecting $N_s$ states sampled from a user-defined distribution $\mathcal{P}$ with support $\mathcal{S} \subseteq \mathcal{X}_N$. We aim to develop a loss function $L_\mathrm{MPC}:\mathbb{R}^{n_x} \times \mathbb{R}^{n_u} \rightarrow \mathbb{R}$ such that the minimization of its expected value taken over $\mathcal{D}_X$ returns a \textit{statistically consistent} estimate of the MPC policy. That is, for any optimal \textit{empirical parameter estimate} satisfying
\begin{equation}\label{eq: IL objective}
    \hat{\theta}_\pi \in \Theta^{\mathcal{D}_X}_\pi:= \arg \min_{\theta_\pi \in \Theta_\pi}\mathbb{E}_{\mathcal{D}_X}[L_\mathrm{MPC}(x,\theta_\pi)]
\end{equation}
and any $x \in \mathcal{S}$, the resulting \textit{empirical policy} $\pi(x,\hat{\theta}_\pi)$ satisfies
\begin{equation}\label{eq: consistency}
    \mathbb{d}\left(\pi(x,\hat{\theta}_\pi),\mathcal{U}_\mathrm{MPC}(x)\right ) \xrightarrow{N_s\rightarrow\infty}0.
    \end{equation}\par
2.) Since practical implementations involve a finite data set size $N_s$ and optimization practices like early stopping and limited hyperparameter tuning, we seek to recover stability guarantees for an \textit{approximate policy} $ \pi(x,\tilde{\theta}_\pi) \approx \pi(x,\hat{\theta}_\pi)$ defined by an \textit{approximate parameter estimate} satisfying
\begin{equation}\label{eq: approximate parameter estimate}
\mathbb{d}\left(\tilde{\theta}_\pi, \Theta^{\mathcal{D}_X}_\pi\right) \approx 0.
\end{equation}\par
Our focus in this work is on NNs as function approximators, though the framework is not limited to them.\par
We assume that the policy NN satisfies the input constraints.
\begin{ass}\label{ass: constraints}
        For all $x \in \mathbb{R}^{n_x}$ and $\theta_\pi \in \Theta_\pi$, it holds that $\pi(x,\theta_\pi) \in \mathbb{U}$.
\end{ass}\par
\cref{ass: constraints} can be satisfied for general polytopic input constraints by, e.g., appending \ifthenelse{\boolean{extendedversion}}{the projection $\min_{u \in \mathbb{U}}  \lVert u - \pi(x,\theta_\pi)\rVert^2$ as a differentiable optimization layer~\cite{OPTNet_Amos} as final layer to the NN.}{a differentiable projection layer~\cite{OPTNet_Amos} as final layer to the NN.}\par
\section{Consistent Imitation Learning for AMPC}\label{sec: approximation procedure}
In this section, we present a tailored IL procedure that enables accurate approximation of~\eqref{eq: generic MPC}. Our focus is on the procedure's mechanism and components. Specifically, \cref{sec: loss function} proposes a tailored loss function, and \cref{sec: IL procedure} subsequently details its integration into a complete IL procedure. The theoretical analysis, where we show that the proposed IL procedure indeed enables statistically consistent approximation of~\eqref{eq: generic MPC} and facilitates the recovery of stability guarantees, is presented in \cref{sec: theoretical analysis}.
\subsection{Loss Function}\label{sec: loss function}
Similar to~\cite{AMPC_Ghezzi}, the main idea is to formulate a loss function that embeds the original intentions of the MPC, and, additionally, remains computationally efficient when optimized via conventional techniques like stochastic gradient descent. 
To this end, we propose the following loss function
\begin{equation}\label{eq: loss function}
    L_\mathrm{MPC}(x,\pi(x,\theta_\pi)) :=\ell(x,\pi(x,\theta_\pi)) + V(f(x,\pi(x,\theta_\pi)),\hat{\theta}_V).
\end{equation}
Equation~\ref{eq: loss function} combines the MPC stage cost $\ell$ and an approximation of the MPC value function $V(x,\hat{\theta}_V) \approx V_\mathrm{MPC}(x)$, and with this the original MPC performance and constraint characteristics, into a computationally tractable loss function.\par %\textcolor{blue}{that integrates smoothly in established machine learning libraries.}\par
The following result provides a system-theoretic justification for ~\eqref{eq: loss function}, establishing ISS for the \textit{optimal} policy
$\pi^*(x)\in \Pi^*(x) :=\arg \min_{u\in \mathbb{U}} L_\mathrm{MPC}(x,u)$
minimizing~\eqref{eq: loss function} point-wise. 
\begin{ass}\label{ass: valfct approximation error} 
    Consider an approximate value function $V(x,\hat{\theta}_V): \mathcal{X}_N\,\cup \,\mathcal{X}_N^+ \times \mathbb{R}^{n_p}\rightarrow \mathbb{R}$ with $\mathcal{X}_N^+:=\{f(x,u)\mid x\in\mathcal{X}_N, \, u\in \mathbb{U} \}$. There exists some $\hat{\varepsilon}_V >0$ such that for all $x \in \mathcal{X}_N$   it holds that $| \varepsilon_V(x) | := | V(x,\hat{\theta}_V) - V_\mathrm{MPC}(x)| \leq \hat{\varepsilon}_V$.
\end{ass}\par
\begin{prop}\label{prop: AOC stability}
    Assume that Assumption~\ref{ass: MPC stability} and~\ref{ass: valfct approximation error} hold. Furthermore, assume that for all $x \in \mathcal{X}_N$ it holds that $f(x,\pi^*(x)) \in \mathcal{X}_N$. 
    Then it follows for all $x(0)\in \mathcal{X}_N$ that application of $u(k) = \pi^*(x(k))$ implies ISS, that is for all $k \geq 0$
    \begin{equation}
        \lVert x(k) \rVert \leq \beta\left(\lVert x(0)\rVert, k\right) + \sigma\left(\sup_{i \in [0,k-1]} | \bar{\varepsilon}_V\left(x(i)\right)|\right)
    \end{equation}
    with $\beta \in \mathcal{KL}$, $\sigma \in \mathcal{K}$, and
    \begin{multline}\label{eq: epsilon hat}
        \bar{\varepsilon}_V(x(k)) := | \varepsilon_V(f(x(k),\pi^*(x(k)))) | \\ + | \varepsilon_V(f(x(k),\pi_\mathrm{MPC}(x(k)))) |.
    \end{multline}
\end{prop}\par
We refer to~\cite[Theorem 3]{AMPC_Nicolas} for the proof.\par
\cref{prop: AOC stability} establishes that by minimizing~\eqref{eq: loss function}, we obtain a policy that guarantees closed-loop stability. 
%
%Moreover, state constraint satisfaction guarantees can be provided assuming a large enough ratio of the slack penalty and the constraint tightening factor $\rho / \eta$~\cite[Theorem 4]{AMPC_Nicolas}. \par
%
\subsection{Imitation Learning Procedure}\label{sec: IL procedure}
Based on~\eqref{eq: loss function}, we propose the two-stage IL procedure summarized in \cref{algo: learning procedure}. % to learn a consistent and stabilizing AMPC policy. 
First, given a user-defined state-value-function data set $\mathcal{D}_V:=\{(x_j,V_\mathrm{MPC}(x_j))\}_{j=1}^{N_s}$ and a value function NN $V(x,\theta_V)$ parameterized by $\theta_V \in \mathbb{R}^{n_p}$, we obtain the approximate value function $V(x,\hat{\theta}_V)$ by applying any suitable regression technique to the data set $\mathcal{D}_V$. Note that, by the continuity and uniqueness of $V_\mathrm{MPC}$, minimizing an MSE-based loss function provides a statistically consistent estimate of $V_\mathrm{MPC}$. Based on this, we use $V(x,\hat{\theta}_V)$ to define~\eqref{eq: loss function} which, together with a user-defined state data set $\mathcal{D}_X$ and policy NN $\pi(x,\theta_\pi)$, defines the IL objective in~\eqref{eq: IL objective}. Minimizing~\eqref{eq: IL objective} then yields the \textit{empirical} policy $\pi(x,\hat{\theta}_\pi)$, which, consequently, approximately minimizes the MPC's objectives, rather than approximately matches the MPC's input mapping.\par
We note that, while we consider an approximation of $V_\mathrm{MPC}$ in~\eqref{eq: loss function} for computational reasons, ideas from~\cite{AMPC_Ghezzi} can be leveraged to implement the exact $V_\mathrm{MPC}$ in~\eqref{eq: loss function}, and, hence, to circumvent approximation errors in~\cref{ass: valfct approximation error}. This, however, introduces significant computational overhead during training, as~\eqref{eq: generic MPC} must be solved for \textit{each} successor state $f(x_j,\pi(x_j,\theta_\pi))$ resulting from \textit{every} $x_j \in \mathcal{D}_X$ in \textit{every} training iteration.\par
\ifthenelse{\boolean{extendedversion}}{
\cref{algo: learning procedure}, similar to the approach in~\cite{AMPC_Ghezzi}, admits conceptual similarities to actor-critic reinforcement learning. Specifically,~\eqref{eq: loss function} can be viewed as an approximate Q-function of the MPC. Consequently, learning $\pi(x,\hat{\theta}_\pi)$ by minimizing~\eqref{eq: loss function} can be interpreted as training the actor NN $\pi(x,\theta_\pi)$ with the Q-function-based critic~\eqref{eq: loss function}.\par}{}
\ifthenelse{\boolean{extendedversion}}{Finally, we note that $\mathcal{D}_X$ and $\mathcal{D}_V$ are not required to collect the same state samples, and that we consider an equal parameter space $\mathbb{R}^{n_p}$ for $\theta_\pi$ and $\theta_V$ for notational simplicity.\par}{}
\begin{algorithm}[t]
    \caption{Proposed learning procedure.}\label{algo: learning procedure}
    \begin{algorithmic}[1]
    \Statex \textbf{Input}: Data sets $\mathcal{D}_V$, $\mathcal{D}_X$, NNs $V(x,\theta_V)$, $\pi(x,\theta_\pi)$. 
    \Statex \textbf{Output}: Empirical policy $\pi(x,\hat{\theta}_\pi)$.\
    
    \State Obtain $V(x,\hat{\theta}_V)$ by solving for
    \begin{equation}
        \hat{\theta}_V \in \arg\min_{\theta_V} \mathbb{E}_{\mathcal{D}_V}[| V(x,\theta_V) - V_\mathrm{MPC}(x)|^2].
    \end{equation}

    \State 
%%%%%%%%%%%%%%%%%%%%%%%%%%%%%%%%%%%%%%%%%%%%
    Construct~\eqref{eq: loss function} and solve for
    \begin{equation}
        \hat{\theta}_\pi \in \arg\min_{\theta_\pi} \mathbb{E}_{\mathcal{D}_X}[L_\mathrm{MPC}(x,\pi(x,\theta_\pi))].
    \end{equation}

    \State Return $ \pi(x,\hat{\theta}_\pi)$.
    \end{algorithmic}
\end{algorithm}\par
\section{Theoretical Analysis}\label{sec: theoretical analysis}
In this section, we analyze the theoretical properties of the IL procedure described in \cref{algo: learning procedure}.
We begin with the statistical consistency properties of the empirical policy $\pi(x,\hat{\theta}_\pi)$, in \cref{sec: statistical properties}. 
We establish that \cref{algo: learning procedure} returns a statistically consistent estimate of the optimal policy $\pi^*(x)$, i.e., ensures that $\pi(x,\hat{\theta}_\pi)$ asymptotically converges to $\Pi^*(x)$ as the number of samples grows to infinity. Based on this, in \cref{sec: stability} we consider the \textit{approximate} policy $\pi(x,\tilde{\theta}_\pi)$, leveraging the results from \cref{sec: statistical properties} to establish ISS guarantees for system~\eqref{eq:nonlinear_system} when controlled by $\pi(x,\tilde{\theta}_\pi)$. \ifthenelse{\boolean{extendedversion}}{See Figure~\ref{fig:theory_schaubild} for a schematic overview of the theoretical connections and results.}{}
\ifthenelse{\boolean{extendedversion}}{
\begin{figure*}[t]
    \centering
    \includegraphics[width=\linewidth]{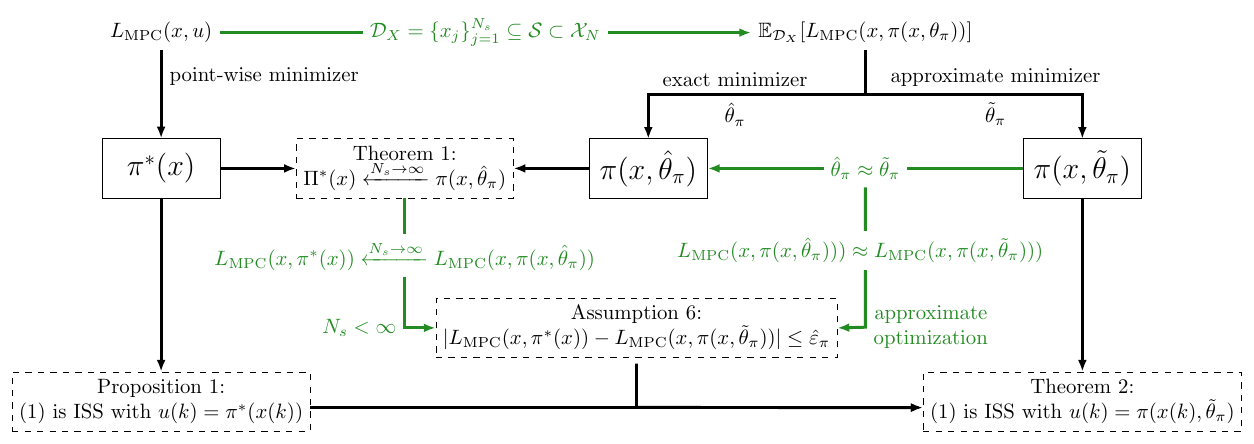} 
    \caption{Relationship (green) of the main theoretical results and definitions (black). Consistency of $\pi(x,\hat{\theta}_\pi)$ (\cref{prop: convergence}) ensures convergence of $\pi(x,\hat{\theta}_\pi)$ to $\Pi^*(x)$ for all $x \in \mathcal{S}$. This implies that $L_\mathrm{MPC}(x,\pi(x,\hat{\theta}_\pi))$ asymptotically approaches the minimal loss $L_\mathrm{MPC}(x,\pi^*(x))$ as the data set size grows. Since furthermore $L_\mathrm{MPC}(x,\pi(x,\hat{\theta}_\pi)) \approx L_\mathrm{MPC}(x,\pi(x,\tilde{\theta}_\pi))$ by the definition of $\pi(x,\tilde{\theta}_\pi)$, this justifies \cref{ass: upper bounds} which accounts for practical conditions like finite data and approximate optimization. Finally, combining \cref{ass: upper bounds} and the ISS guarantees for $\pi^*(x)$ (\cref{prop: AOC stability}) enables us to recover ISS guarantees for $\pi(x,\tilde{\theta}_\pi)$ (\cref{prop: iss}).} 
    \label{fig:theory_schaubild}
\end{figure*}}{}
\subsection{Statistical Consistency}\label{sec: statistical properties}
To establish statistical consistency, we take the following statistical view-point. We consider each sampled state $x_j$ as a random vector drawn from the user-defined probability distribution $\mathcal{P}$ with probability density function $p(x)$. We assume $\mathcal{P}$ is well-defined, meaning $p(x) > 0$ for all $x$ in its compact support $\mathcal{S}$, and $\int_{\mathcal{X}_N}p(x)dx = \int_\mathcal{S} p(x)dx = 1$. The data set $\mathcal{D}_X = \{x_j\}_{j=1}^{N_s}$ is then given as a sequence of independent and identically distributed (iid) random vectors $x_j$, $j = 1, \ldots, N_s$, with length $N_s$. Hence, for each $N_s \in \mathbb{N}$, $\mathcal{D}_X$ can be viewed as a random variable. Consequently, every quantity derived from a measurable mapping acting on $\mathcal{D}_X$ remains a random variable, and depends on $N_s$. Specifically, any empirical policy $\pi(x,\hat{\theta}_\pi)$ with $\hat{\theta}_\pi \in \Theta^{\mathcal{D}_X}_\pi$ is a random variable depending on $N_s$. Our focus is on analyzing the sequence of any $\pi(x,\hat{\theta}_\pi)$ as $N_s$ increases to infinity. For notational clarity, however, we omit the explicit dependence on $N_s$ in the following.
Furthermore, we ensure measurability of $\pi(x,\theta_\pi)$, $L_\mathrm{MPC}$, and related quantities with the following assumption.
\begin{ass}\label{ass: continuity}
        $V(x,\hat{\theta}_V)$ is continuous in~$x$, and $\pi(x,\theta_\pi)$ is continuous in $x$ and $\theta_\pi$. 
\end{ass}\par
\cref{ass: continuity} is satisfied by common NN architectures like feedforward NNs with ReLU activation functions. Nonetheless, requiring continuity of $\pi(x,\theta_\pi)$ with respect to $x$ restricts the scope of this section to settings where at least one continuous policy $\pi^*(x)$ exists. This is the case, e.g., if $J_\mathrm{P}$ is quadratic and \eqref{eq:nonlinear_system} is linear.\par
Finally, we require the policy NN $\pi(x,\theta_\pi)$ to be sufficiently expressive to approximate at least one optimal policy $\pi^*(x)$ over $\mathcal{S}$. To this end, we define the set of \textit{optimal parameters} by
\begin{equation}
\Theta^*_\pi(\mathcal{S}):=\{\theta_\pi \in \Theta_\pi \mid  \forall x \in \mathcal{S}: \pi(x,\theta_\pi) \in \Pi^*(x)\},
\end{equation}
and assume the following.
\begin{ass}\label{ass: existence ideal paramter}
    It holds that $\Theta^*_\pi(\mathcal{S}) \neq \emptyset$.
\end{ass}\par
\begin{theo}\label{prop: convergence}
    Suppose Assumptions~\ref{ass: continuity} and~\ref{ass: existence ideal paramter} hold.
    Then, for all $x \in \mathcal{S}$ and any $\pi(x,\hat{\theta}_\pi)$ with $\hat{\theta}_\pi \in  \Theta^{\mathcal{D}_X}_\pi$ it holds with probability 1 that 
    \begin{equation}\label{eq: convergence result}
        \mathbb{d}(\pi(x,\hat{\theta}_\pi),\Pi^*(x)) \xrightarrow{N_s\rightarrow \infty} 0.
    \end{equation}
\end{theo}\par 
\ifthenelse{\boolean{extendedversion}}{See Appendix~\ref{sec: proof convergence} for the proof.\par}
    {\begin{proof}
    In the following, we outline the key steps for proving \cref{prop: convergence}, while the complete proof can be found in~\cite{AMPC_Milios_Statistical_L-CSS_Extended}.  The proof builds upon four main components:\par
    1) \textit{Convergence of objectives:} With probability 1 it holds that
    \begin{multline}\label{eq: uniform convergence of risks}
        \sup_{\theta_\pi \in \Theta_\pi} |\mathbb{E}_{\mathcal{D}_X}[L_\mathrm{MPC}(x,\pi(x,\theta_\pi))] \\ - \mathbb{E}_\mathcal{S}[L_\mathrm{MPC}(x,\pi(x,\theta_\pi))]| \xrightarrow{N_s \rightarrow \infty} 0.
    \end{multline}
    This follows from~\cite[Theorem 7.48]{stochastic_programming}, which applies due to continuity of $L_\mathrm{MPC}$ and compactness of $\mathcal{S}$ implying that $L_\mathrm{MPC}$ is bounded by an integrable function~\cite[Sec. 7.2]{stochastic_programming}.\par
    2) \textit{Convergence of optimal empirical and expected parameters:} With probability 1 it holds that $\mathbb{D}(\Theta^{\mathcal{D}_X}_\pi,\Theta^\mathcal{S}_\pi) \xrightarrow[]{N_s \rightarrow \infty} 0$, where $\Theta^\mathcal{S}_\pi:= \arg \min_{\theta_\pi \in \Theta_\pi} \mathbb{E}_\mathcal{S} [L_\mathrm{MPC}(x,\pi(x,\theta_\pi))]$. This follows from~\cite[Theorem 5.3]{stochastic_programming}, which applies due to \cref{ass: existence ideal paramter}, compactness of $\Theta_\pi$, and since by~\cite[Theorem 7.48]{stochastic_programming} $\mathbb{E}_\mathcal{S}[L_\mathrm{MPC}(x,\pi(x,\theta_\pi))]$ is finite valued and continuous on $\Theta_\pi$.\par
    3) \textit{Equivalence of expected and optimal parameters:} It holds that $\Theta^*_\pi(\mathcal{S}) = \Theta^\mathcal{S}_\pi$, implying that $\mathbb{D}(\Theta^{\mathcal{D}_X}_\pi,\Theta^*_\pi(\mathcal{S}))\xrightarrow[]{N_s \rightarrow \infty} 0$ with probability 1. This follows by monotonicity of the $\mathbb{E}$ operator together with the fact that $p(x)>0$ for all $x \in \mathcal{S}$.\par
    4) \textit{Convergence of learned policies:} Finally, \eqref{eq: convergence result} follows from 3.) due to uniform continuity of $\pi(x,\theta_\pi)$, which concludes the proof.
    \end{proof}\par}
We note that \cref{prop: convergence} generally holds when \cref{ass: continuity} and~\ref{ass: existence ideal paramter} are satisfied and consequently applies beyond the SCMPC formulation in~\eqref{eq: generic MPC}.
\subsection{Input-to-State Stability}\label{sec: stability}
In this section, we are interested in recovering ISS guarantees for an approximate policy $\pi(x,\tilde{\theta}_\pi) \approx \pi(x,\hat{\theta}_\pi)$.
As a first step, 
\ifthenelse{\boolean{extendedversion}}{we relate the loss obtained by $\pi(x,\tilde{\theta}_\pi)$ with the loss obtained by the optimal policy $\pi^*(x)$.}{we assume the following.}
\begin{ass}\label{ass: upper bounds}
    Consider an approximate policy $\pi(x,\tilde{\theta}_\pi)$ with $\tilde{\theta}_\pi$ satisfying~\eqref{eq: approximate parameter estimate}. We assume that there exists $\hat{\varepsilon}_\pi >0$ such that for all $x \in \mathcal{X}_N$ and any $\pi^*(x) \in \Pi^*(x)$ it holds that
    \begin{equation}    
        \ | \varepsilon_\pi(x) | := | L_\mathrm{MPC}(x,\pi(x,\tilde{\theta}_\pi)) - L_\mathrm{MPC}(x,\pi^*(x)) | \leq \hat{\varepsilon}_\pi.
    \end{equation}\par
\end{ass}\par
\cref{ass: upper bounds}, and also \cref{ass: valfct approximation error}, can be verified, e.g., using the approach in~\cite[Sec. IV]{AMPC_Hertneck}. Importantly, \cref{prop: convergence} guarantees that~\cref{ass: upper bounds} can, in principle, be satisfied for arbitrary small $\hat{\varepsilon}_\pi$. Specifically, it ensures that there exists a data set size $N_s$ such that for all $\hat{\theta}_\pi \in \Theta^{\mathcal{D}_X}_\pi$, $\pi^*(x) \in \Pi^*(x)$, and $x \in \mathcal{S}$, the resulting empirical policy $\pi(x,\hat{\theta}_\pi)$ satisfies $| L_\mathrm{MPC}(x,\pi(x,\hat{\theta}_\pi))- L_\mathrm{MPC}(x,\pi^*(x))| \leq \hat{\varepsilon}_\pi$ with probability~1. 
Consequently, \cref{ass: upper bounds} translates into requiring $N_s$ to be large enough and $\tilde{\theta}_\pi$ to be a good approximation of $\hat{\theta}_\pi$.\par
Based on this, we can relate the loss of $\pi(x,\tilde{\theta}_\pi)$ to the cost resulting from the original MPC policy.
\begin{lem}\label{lem: V relation}
    Suppose the conditions from \cref{prop: AOC stability} and Assumption~\ref{ass: upper bounds} hold. Assume that $\pi(x,\tilde{\theta}_\pi)$ is such that for all $x \in \mathcal{X}_N$ it holds that $f(x,\pi(x,\tilde{\theta}_\pi)) \in \mathcal{X}_N$. Then, for all $x \in \mathcal{X}_N$, it holds that
    \begin{multline}\label{eq: loss mpc cost relation}
        L_\mathrm{MPC}(x,\pi(x,\tilde{\theta}_\pi)) \\ \leq \ell(x,\pi_\mathrm{MPC}(x))  + V_\mathrm{MPC}(f(x,\pi_\mathrm{MPC}(x))) \\ + | \varepsilon_V(f(x,\pi_\mathrm{MPC}(x))) |  + | \varepsilon_\pi(x)|.
    \end{multline}
\end{lem}
\begin{proof}
    By Assumption~\ref{ass: upper bounds}, for any $\pi^*(x) \in \Pi^*(x)$ and all $x \in \mathcal{X}_N$, we get that
    $L_\mathrm{MPC}(x,\pi(x,\tilde{\theta}_\pi)) \leq L_\mathrm{MPC}(x,\pi^*(x)) + | \varepsilon_\pi(x)|$.
    Furthermore, based on \cref{prop: AOC stability}, we can apply \cite[Lemma 2]{AMPC_Nicolas} to obtain
    \begin{multline}\label{eq: V relation}
        L_\mathrm{MPC}(x,\pi^*(x)) \leq \ell(x,\pi_\mathrm{MPC}(x)) + \\ V_\mathrm{MPC}(f(x,\pi_\mathrm{MPC}(x))) + |\varepsilon_V(f(x,\pi_\mathrm{MPC}(x)))|
    \end{multline}
    holding for all $x \in \mathcal{X}_N$ and any $\pi^*(x) \in \Pi^*(x)$. Consequently, we obtain~\eqref{eq: loss mpc cost relation} by combining~\eqref{eq: V relation} with the result above, which concludes the proof.
\end{proof}\par
With this, we can establish ISS, ensuring stability despite approximation errors.
\begin{theo}\label{prop: iss}
    Suppose the conditions from Lemma~\ref{lem: V relation} hold. Then it follows that for all $x(0) \in \mathcal{X}_N$ application of $u(k) = \pi(x(k),\tilde{\theta}_\pi)$ implies ISS with respect to $\bar{\varepsilon}_\pi(x(k)) :=\bar{\varepsilon}_V(x(k)) + |\varepsilon_\pi(x(k))|$.
\end{theo}
\begin{proof}
    By assumption, the set $\mathcal{X}_N$ is invariant under the application of $\pi(x,\tilde{\theta}_\pi)$. Furthermore, by Assumption \ref{ass: upper bounds}, we have that $\bar{\varepsilon}_\pi$ lies in the compact set $\bar{\varepsilon}_\pi \in [0,2\hat{\varepsilon}_V+\hat{\varepsilon}_\pi]$ which contains the origin. This together with Lemma~\ref{lem: V relation} enables to apply the arguments used in \cite[Theorem 3]{AMPC_Nicolas}, establishing the desired result.
\end{proof}\par
\ifthenelse{\boolean{extendedversion}}{
An important consequence of Proposition~\ref{prop: iss} is that if no approximation errors are obtained, i.e., $V(x,\hat{\theta}_V)= V_\mathrm{MPC}(x)$ and $\pi(x,\tilde{\theta}_\pi) = \pi^*(x)$, asymptotic stability of the origin of system~\eqref{eq:nonlinear_system} is recovered. To this end, the approach of~\cite{AMPC_Ghezzi} can be leveraged to circumvent value function approximation errors.\par}{}
We note that state constraint satisfaction guarantees under application of $\pi(x,\tilde{\theta}_\pi)$ can be established as well, using similar arguments as in~\cite[Theorem 4]{AMPC_Nicolas}, and additionally assuming that $\hat{\varepsilon}_\pi \leq \hat{\varepsilon}_V$.\par

\section{Numerical Experiments}\label{sec: numerical experiments}
In this section, we demonstrate the effectiveness of our approach with an academic example in \cref{sec: motivating example}, and a 2D robot obstacle avoidance scenario in \cref{sec: robot obstacle avoidance example}.
\subsection{Illustrative Example}\label{sec: motivating example}
\begin{figure}[b]
    \centering
    \includegraphics[width=\linewidth]{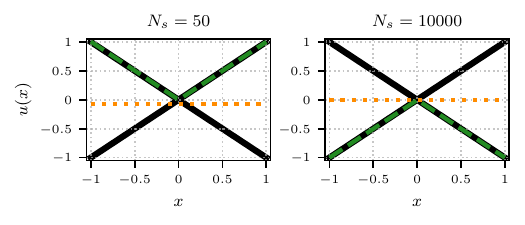}
    \caption{Comparison of the set-valued MPC mapping $\Pi^*(x)$ (solid black) with the trained BC policy $\pi(x,\theta^\mathrm{BC}_\pi)$ (dotted orange) and the approximate policy $\pi(x,\tilde{\theta}_\pi)$ (dashed green) trained with \cref{algo: learning procedure}.}
    \label{fig:example}
\end{figure}
 We consider the following MPC problem~\cite{AMPC_Cao}:
\begin{equation}\label{eq: example1}
    \min_{u_{0|k}}  \sum_{i=0}^{1} x_{i|k}^2 \ \text{s.t. }  \  x_{0|k} = x(k), \  x_{1|k} = x_{0|k}^2-u_{0|k}^2.
\end{equation}
For any state $x \in \mathbb{R}$, the MPC input is given by $\pi_\mathrm{MPC}(x) \in \Pi^*(x) = \{x,-x\}$.
We generate two state-input data sets $\mathcal{D}_\pi:=\{(x_j,\pi_\mathrm{MPC}(x_j))\}_{j=1}^{N_s}$ with $N_s=50$ and $N_s= 10000$, respectively, by sampling states uniformly from the interval $[-1,1]$ and solving~\eqref{eq: example1} for each sampled state. The solver initialization is chosen randomly between $\{-x_j,x_j\}$.\par
We construct~\eqref{eq: loss function} by employing the exact MPC value function $V(x,\hat{\theta}_V) = V_\mathrm{MPC}(x) = x^2$ for simplicity. The policy NN $\pi(x,\theta_\pi)$ is defined by one linear hidden layer of two neurons and ReLU activation functions, which is sufficient to replicate all continuous selections of $\pi_\mathrm{MPC}$. We obtain two approximate policies $ \pi(x,\theta^\mathrm{BC}_\pi)$ and $\pi(x,\tilde{\theta}_\pi)$, where we apply BC for $\theta^\mathrm{BC}_\pi \approx \arg \min_{\theta_\pi} \mathbb{E}_{\mathcal{D}_\pi} \lVert \pi(x,\theta_\pi)-\pi_\mathrm{MPC}(x)\rVert^2$, and \cref{algo: learning procedure} for $\tilde{\theta}_\pi\approx \arg \min_{\theta_\pi} \mathbb{E}_{\mathcal{D}_\pi}[L_\mathrm{MPC}(x,\pi(x,\theta_\pi))]$.
\ifthenelse{\boolean{extendedversion}}{
In both cases, we deploy stochastic gradient descent on mini batches over 2000 epochs using the ADAM optimizer~\cite{ADAM}.}
{}\par 
\cref{fig:example} compares the original MPC mapping with both approximate policies. Despite the NN's potential to represent the MPC, the BC policy $\pi(x,\theta^\mathrm{BC}_\pi)$ deviates considerably from the set-valued MPC control law. In fact, as expected, we empirically validate that for each $x$, the BC policy returns $\pi(x,\theta^\mathrm{BC}_\pi) \approx \int_{\Pi^*(x)} u \: p_\mathrm{MPC} ( u \mid x) d u = 0$, where $p_{\mathrm{MPC}}(u  \mid x)\approx 0.5$ for each $u \in \Pi^*(x)$ and zero elsewhere. Consequently, for any interval $x \in I=[a,b]$ with $a\leq b$, as $N_s\rightarrow \infty$, we observe that $\mathbb{d}(\pi(x,\theta^\mathrm{BC}_\pi), \Pi^*(x)) \rightarrow x$ and $\mathbb{E}_{I}[\mathbb{d}(\pi(x,\theta^\mathrm{BC}_\pi),\Pi^*(x))] \rightarrow \frac{b^2-a^2}{2}$.\par
In contrast, the approximate policy $\pi(x,\tilde{\theta}_\pi)$ closely aligns with the MPC. Furthermore, while $\mathbb{d}(\pi(x,\tilde{\theta}_\pi),{\Pi^*(x)})$ is small even for $N_s=50$, we find that, as the number of state samples increases, the approximate policy converges to one of the MPC policies, i.e., $\mathbb{d}(\pi(x,\tilde{\theta}_\pi),{\Pi^*(x)}) \xrightarrow{N_s\rightarrow \infty} 0$.
\subsection{2D Robot Obstacle Avoidance}\label{sec: robot obstacle avoidance example}
\begin{figure}[t]
    \centering
    \includegraphics[width=\linewidth]{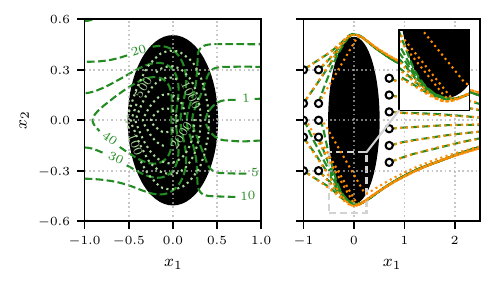}
    \caption{\textbf{Left}: Sub-level sets of $V(x,\hat{\theta}_{V}^\mathrm{p})$ (dashed dark green) and $V(x,\hat{\theta}_{V}^\xi)$ (dotted light green). \textbf{Right:} Closed-loop trajectories resulting from $\pi(x,\theta^\mathrm{BC}_\pi)$ (dotted orange) and $\pi(x,\tilde{\theta}_\pi)$ (dashed green), respectively, with the starting points indicated by white circles. The obstacle is illustrated by the black area. }
    \label{fig: 2d_robot}
\end{figure}
We consider a simplified unicycle model defined by
$x(k+1) = x(k) + 0.05 \begin{bmatrix} \cos(u(k))\\sin(u(k))\end{bmatrix},$
simulating, e.g., a robot moving in the 2-D plane with constant speed. Here, $x = (x_1,x_2) \in \mathbb{R}^2$ denotes the robot's position, and $u \in [-\frac{\pi}{3}, \frac{\pi}{3}]$ directly controls its orientation. The control objective is to circumvent a circular obstacle $\mathcal{O}:=\{x\in \mathbb{R}^2 \mid x^\top x\leq 0.5^2\}$,
while converging to the center line $\mathcal{C}:=\{x=(x_1,x_2) \in \mathbb{R}^2 \mid x_2=0\}$.\par
We formulate the MPC problem according to~\eqref{eq: generic MPC}, with~\eqref{eq: slacked state constraints} adapted to  $(x_{i|k})^\top x_{i|k} \geq 0.5^2+\eta + \xi_{i|k}+\xi_{N|k}$, and~\eqref{eq: term constraint}~-~\eqref{eq: scaled terminal set} removed. We choose $\ell(x,u) = \lVert x\rVert^2_Q+\lVert u \rVert_R^2$, $V_\mathrm{f}(x)= \lVert x \Vert^2_{Q_N}$, with $Q = \left[\begin{smallmatrix} 0 & 0\\ 0& 1 \end{smallmatrix} \right]$, $R = 5$, $Q_N = 100Q$, and $\ell_\xi(\xi)=\lVert \xi\rVert_1$, $\rho = 15000$, and $\eta = 0.01$.
\ifthenelse{\boolean{extendedversion}}{
We define a coarse two-dimensional grid within the interval $x_1 \in [-2,2]$ and $x_2 \in [-1.5,1.5]$ using 41 equidistant points along each axis, and a refined grid within the interval $x_1 \in [-1.5,1.5]$ and $x_2 \in [-0.5,0.5]$ using 51 and 31 equidistant points, respectively. Based on this grid, we generate a data set $\mathcal{D}:=\{(x_j,\pi_\mathrm{MPC}(x_j),V_\mathrm{MPC}(x_j)=V_\mathrm{MPC}^{\mathrm{p}}(x_j)+V_\mathrm{MPC}^\xi(x_j))\}_{j=1}^{N_s}$ consisting of $N_s=3262$ samples in total.}{We obtain a data set $\mathcal{D}:=\{(x_j,\pi_\mathrm{MPC}(x_j),V_\mathrm{MPC}(x_j)=V_\mathrm{MPC}^{\mathrm{p}}(x_j)+V_\mathrm{MPC}^\xi(x_j))\}_{j=1}^{N_s}$ via gridding.}
We obtain
$V(x,\hat{\theta}_V):=\max(0,V(x,\hat{\theta}_\mathrm{V}^\mathrm{p})) + \max(0,V(x,\hat{\theta}_{V}^\xi)) \approx V_\mathrm{MPC}(x)$
and $\pi_\mathrm{BC}$
with
$\hat{\theta}_{V}^\mathrm{p}\approx \arg\min_{\theta_V} \mathbb{E}_{\mathcal{D}}[\lVert V(x,\theta_V)-V^\mathrm{p}_\mathrm{MPC}(x)\rVert^2]$,
$\hat{\theta}_{V}^\mathrm{\xi}\approx \arg\min_{\theta_V} \mathbb{E}_{\mathcal{D}}[\lVert V(x,\theta_\mathrm{V})-V^\mathrm{\xi}_\mathrm{MPC}(x)\rVert^2]$,
and $\theta_\pi^\mathrm{BC} \approx \arg \min_{\theta_\pi} \mathbb{E}_{\mathcal{D}} \lVert \pi(x,\theta_\pi) - \pi_\mathrm{MPC}(x) \rVert ^2$, using the method in~\cite[Sec. 4.3]{AMPC_Nicolas} for efficient value function approximation. We estimate an absolute/relative value function approximation error bound $\hat{\varepsilon}_V \approx \max_{x \in \mathcal{D}}\varepsilon_V(x)$ of $ 93.7 /0.7$. Based on this, we apply \cref{algo: learning procedure} to obtain $\pi(x,\tilde{\theta}_\pi)$ with $\tilde{\theta}_\pi\approx \arg \min_{\theta_\pi} \mathbb{E}_{\mathcal{D}}[L_\mathrm{MPC}(x,\pi(x,\theta_\pi))]$. 
\ifthenelse{\boolean{extendedversion}}{
All the considered NNs consist of three hidden ReLU layers with 128 neurons each. We use stochastic gradient descent for the minimization, where we deploy the ADAM optimizer~\cite{ADAM} for 2000 epochs and apply gridsearch, varying the learning and the exponential learning decay rate. The considered parameter grids are $[5\times 10^{-5},10^{-4}, 5\times 10^{-4}, 10^{-3}, 5\times 10^{-3}, 10^{-2}]$ for the learning and $[0.9995, 0.999, 0.995, 0.99]$ for the exponential decay rate.\par
}{
We use stochastic gradient descent for the minimization, and all NNs consist of three hidden ReLU layers with 128 neurons each. For the detailed learning setup and MPC implementation see~\cite{AMPC_Milios_Statistical_L-CSS_Extended}.\par
}
\cref{fig: 2d_robot} shows the learned performance and constraint value functions $V(x,\hat{\theta}_{V}^\mathrm{p})$ and $V(x;\hat{\theta}_{V}^\xi)$, respectively, and 16 closed-loop trajectories resulting from $\pi(x,\tilde{\theta}_\pi)$ and $\pi(x,\theta^\mathrm{BC}_\pi)$. $V(x,\hat{\theta}_{V}^\mathrm{p})$ successfully represents the MPC's performance characteristics, while $V(x,\hat{\theta}_{V}^\xi)$ enables constraint-aware learning of $\pi(x,\tilde{\theta}_\pi)$ by successfully embedding constraint information through high cost values within $\mathcal{O}$. Importantly, the SCMPC's ability to gather data within $\mathcal{O}$ is essential for this. We observe that $\pi(x,\tilde{\theta}_\pi)$ indeed successfully circumvents $\mathcal{O}$ and converges to $\mathcal{C}$, while the BC policy $\pi_\mathrm{BC}$ leads to significant constraint violations.\par
In \cref{tab:comparison}, we compare the average performance and policy evaluation time, and the number of constraint violations obtained by $\pi(x,\tilde{\theta}_\pi)$, $\pi(x,\theta^\mathrm{BC}_\pi)$, and the optimal policy $\pi^*(x)$ over $500$ trajectories of length $100$, where the initial conditions are sampled uniformly in the interval $x_1 \in [-1,0]$, $x_2 \in [-0.7,0.7]$ and discarded if $x \in \mathcal{O}$. \ifthenelse{\boolean{extendedversion}}{Note that this may still include initial states close to $\mathcal{O}$ for which state constraint violations are unavoidable.}{} Evaluation of $\pi^*(x)$ is done using a uniform input gridding with 100 points. The performance of a trajectory is defined as the sum of a tracking part $p_\mathrm{t}=\sum_{k=0}^{99}\ell(x(k),u(k))$ and a state constraint part $p_\mathrm{c}=\sum_{k=0}^{99}\max(0,V(x(k),\hat{\theta}^\xi_V))$. The number of constraint violations is the total number of visited states that lie within $\mathcal{O}$. We estimate an absolute/relative policy approximation error bound $\hat{\varepsilon}_\pi \approx \max_{x \in \mathcal{D}} \varepsilon_\pi(x)$ of $144.1/0.4$.
\ifthenelse{\boolean{extendedversion}}{Our setup included an Intel Core i7-13850HX CPU, 30 GB of RAM, running Ubuntu 22.04.}{}
\ifthenelse{\boolean{extendedversion}}{We observe that $\pi^*(x)$ achieves the best performance with little constraint violations. The low constraint part $p_\mathrm{c}$ of the performance metric indicates that the constraint violations, while higher in number compared to $\pi(x,\tilde{\theta}_\pi)$, are less severe. Nevertheless, $\pi(x,\theta^\mathrm{BC}_\pi)$ approximates the behavior of $\pi^*(x)$ well, achieving similar performance and number of constraint violations. In contrast, $\pi(x,\theta^\mathrm{BC}_\pi)$ results in poor performance and significant constraint violations. Notably, both $\pi(x,\tilde{\theta}_\pi)$ and $\pi(x,\theta^\mathrm{BC}_\pi)$ incur substantially smaller evaluation time compared to $\pi^*(x)$.}
{We observe that $\pi(x,\theta^\mathrm{BC}_\pi)$ admits poor performance and significant constraint violations. Conversely, both policies $\pi(x,\tilde{\theta}_\pi)$ and $\pi^*(x)$ achieve good performance with little constraint violations. However, $\pi^*(x)$ incurs substantially higher evaluation time.}
\begin{table}[t]
\centering
\caption{Quantitative comparison of the policies.}
\label{tab:comparison}
\begin{tabular}{@{} l c c c @{}} % Policy column (left), P, C, T (centered)
\toprule
Policy & Perf. $(p_\mathrm{t}+p_\mathrm{c})$ & Eval. Time (s) & Nb. Constr. Viol. \\
\midrule
$\pi(x,\theta^\mathrm{BC}_\pi)$ & $3.6 \: (0.21+3.4)$ & \num{3.8e-6} & 861 \\
$\pi(x,\tilde{\theta}_\pi)$ & $0.31 \: (0.23 + 0.09)$ & \num{3.8e-6} & 108 \\
$\pi^*(x)$ & $0.24 \: (0.22 + 0.02)$ & \num{2.5e-3} & 185 \\
\bottomrule
\end{tabular}
\end{table}

\bibliographystyle{IEEEtran}
\bibliography{References}

\ifthenelse{\boolean{extendedversion}}{\appendix 
\subsection{Proof of \cref{prop: convergence}}\label{sec: proof convergence}}{}
\ifthenelse{\boolean{extendedversion}}{To ensure a valid statistical argumentation, we first establish measurability of $\pi(x,\theta_\pi)$, the \textit{empirical} IL objective $\mathbb{E}_{\mathcal{D}_X}[L_\mathrm{MPC}(x,\pi(x,\theta_\pi))]$, and the set of optimal empirical parameter estimates $\Theta^{\mathcal{D}_X}_\pi$. Furthermore, we establish that the \textit{expected} IL objective $\mathbb{E}_\mathcal{S}[L_\mathrm{MPC}(x,\pi(x,\theta_\pi))]$ is finite valued and well-defined. Based on this, we prove Proposition~\ref{prop: convergence} by linking $\Theta^{\mathcal{D}_X}_\pi$ with $\Theta^*_\pi(\mathcal{S})$ via the set of expected parameter estimates $\Theta^\mathcal{S}_\pi:= \arg \min_{\theta_\pi \in \Theta_\pi} \mathbb{E}_\mathcal{S} [L_\mathrm{MPC}(x,\pi(x,\theta_\pi))]$. This we do in 5 steps:
    \begin{enumerate}
        \item First, we show that the empirical IL objective converges to the expected IL objective uniformly in $\theta_\pi$. That is, with probability 1 it holds that
        \begin{multline}\label{eq: uniform convergence of risks}
            \sup_{\theta_\pi \in \Theta_\pi} |\mathbb{E}_{\mathcal{D}_X}[L_\mathrm{MPC}(x,\pi(x,\theta_\pi))] \\ - \mathbb{E}_\mathcal{S}[L_\mathrm{MPC}(x,\pi(x,\theta_\pi))]| \rightarrow 0
        \end{multline}
        as $N_s \rightarrow \infty$.
        \item Second, we show that the set of optimal empirical parameter estimates converges to the set of expected parameter estimates. That is, we show that with probability 1 it holds that
        \begin{equation}\label{eq: convergence of parameters}
            \mathbb{D}(\Theta^{\mathcal{D}_X}_\pi,\Theta^\mathcal{S}_\pi) \xrightarrow[]{N_s \rightarrow \infty} 0.
        \end{equation}
        \item Next, we link the set of expected parameter estimates with the set of optimal parameters by deriving that
        \begin{equation}\label{eq: equality of parameters}
             \Theta^\mathcal{S}_\pi = \Theta^*_\pi(\mathcal{S}).
        \end{equation}
        \item Based on this, we can link $\Theta^{\mathcal{D}_X}_\pi$ with $\Theta^*_\pi$, showing that with probability 1 it holds that
        \begin{equation}
            \mathbb{D}(\Theta^{\mathcal{D}_X}_\pi,\Theta^*_\pi(\mathcal{S}))\xrightarrow[]{N_s \rightarrow \infty} 0.
        \end{equation}
        \item Finally, we conclude that step 4. implies the desired result. That is that for any $\hat{\theta}_\pi \in  \Theta^{\mathcal{D}_X}_\pi$ and all $x \in \mathcal{S}$, with probability 1 it holds that 
        \begin{equation}
            \mathbb{d}(\pi(x,\hat{\theta}_\pi),\Pi^*(x) \xrightarrow[]{N_s \rightarrow \infty} 0.
        \end{equation}
    \end{enumerate}\par
    \textit{Measurability:}
    By \cref{ass: continuity}, that is, continuity, $\pi(x,\theta_\pi)$ and  $V(x,\hat{\theta}_V)$ are measurable with respect to $x$. Consequently, we have that $L_\mathrm{MPC}(x,\pi(x,\theta_\pi))$ is continuous and measurable with respect to $x$.  Furthermore, since by Assumption~\ref{ass: continuity} $\pi(x,\theta_\pi)$ and therefore $L_\mathrm{MPC}(x,\pi(x,\theta_\pi))$ are continuous in $\theta_\pi$, we have that $L_\mathrm{MPC}(x,\pi(x,\theta_\pi))$ is a random lower semi-continuous function~\cite[Definition 7.35]{stochastic_programming}. This allows to apply \cite[Theorem 7.37]{stochastic_programming}, yielding that $\mathbb{E}_{\mathcal{D}_X}[L_\mathrm{MPC}(x,\pi(x,\theta_\pi))]$ and $\Theta^{\mathcal{D}_X}_\pi$ are measurable.\par
    \textit{Well-definedness:} Well-definedness of $\mathbb{E}_\mathcal{S}[L_\mathrm{MPC}(x,\pi(x,\theta_\pi))]$ follows from the distribution  $\mathcal{P}$ generating $\mathcal{D}_X$ being well-defined. Finally, to conclude that $\mathbb{E}_\mathcal{S}[L_\mathrm{MPC}(x,\pi(x,\theta_\pi))]$ is finite valued, we use compactness of $\mathcal{S}$ and $\theta_\pi$ together with continuity of $L_\mathrm{MPC}(x,\pi(x,\theta_\pi))$, implying that $L_\mathrm{MPC}(x,\pi(x,\theta_\pi))$ remains bounded for all $x \in \mathcal{S}$ and all $\theta_\pi \in \Theta_\pi$. Since, to obtain the expected IL objective $\mathbb{E}_\mathcal{S}[L_\mathrm{MPC}(x,\pi(x,\theta_\pi))]$, we integrate over the compact domain $\mathcal{S}$, boundedness of $L_\mathrm{MPC}(x,\pi(x,\theta_\pi))$ together with well-definedness of $\mathcal{P}$ implies that $\mathbb{E}_\mathcal{S}[L_\mathrm{MPC}(x,\pi(x,\theta_\pi))]$ is bounded as well.\par
    1) Having measurability and well-definedness established, we continue with deriving \eqref{eq: uniform convergence of risks} by relying on \cite[Theorem 7.48]{stochastic_programming}. To this end, we require that:
    \begin{enumerate}
        \item[(i)] $L_\mathrm{MPC}(x,\pi(x,\theta_\pi))$ is continuous in $x$ and $\theta_\pi$;
        \item[(ii)] $L_\mathrm{MPC}(x,\pi(x,\theta_\pi))$ is dominated by an integrable function;
        \item[(iii)] The samples $x_j$, $j=1, \ldots, N_s$ are iid.
    \end{enumerate}
    Condition (i) holds by Assumption~\ref{ass: continuity}. Similarly, condition (iii) holds by assumption. Hence, it remains to show that (ii) holds, for which we again use boundedness of $L_\mathrm{MPC}(x,\pi(x,\theta_\pi))$. More specifically, boundedness of $L_\mathrm{MPC}(x,\pi(x,\theta_\pi))$ implies that there exists $M>0$ such that for all $x \in \mathcal{S}$ and $\theta_\pi \in \Theta_\pi$ it holds $|L_\mathrm{MPC}(x,\pi(x,\theta_\pi))|\leq M$.
    Consequently, defining the measurable function  $h(x) = M$, we can conclude that for all $x \in \mathcal{S}$ it holds that $L_\mathrm{MPC}(x,\pi(x,\theta_\pi)) \leq h(x)$ and $\int_S h(x)p(x)dx = M < \infty$, and, hence, that $L_\mathrm{MPC}(x,\pi(x,\theta_\pi))$ is indeed bounded by the integrable function $h(x)$. Ultimately, we can apply \cite[Theorem 7.48]{stochastic_programming}, yielding \eqref{eq: uniform convergence of risks}.\par
    2) To establish \eqref{eq: convergence of parameters}, we aim to apply \cite[Theorem 5.3]{stochastic_programming}. To this end, we require a compact set $C \subset \mathbb{R}^{n_p}$ such that:
    \begin{enumerate}
        \item[(iv)]  $\Theta^*_\pi(\mathcal{S})$ is not empty and contained in $C$;
        \item[(v)] $\mathbb{E}_\mathcal{S}[L_\mathrm{MPC}(x,\pi(x,\theta_\pi))]$ is finite valued and continuous on $C$;
        \item[(vi)] $\mathbb{E}_{\mathcal{D}_X}[L_\mathrm{MPC}(x,\pi(x,\theta_\pi))]$ converges to $\mathbb{E}_\mathcal{S}[L_\mathrm{MPC}(x,\pi(x,\theta_\pi))]$ with probability 1 as $N_s \rightarrow \infty$ uniformly in $\theta_\pi \in C$;
        \item[(vii)] With probability 1 for $N_s$ large enough the set of empirical parameter estimates is nonempty and contained in $C$.
    \end{enumerate}
    We show that all the requirements are met for the compact parameter set $\Theta_\pi \subset \mathbb{R}^{n_p}$.
    While (iv) holds true by Assumption~\ref{ass: existence ideal paramter}, for
    (v), we again rely on \cite[Theorem 7.48]{stochastic_programming}, which provides that $\mathbb{E}_\mathcal{S}[L_\mathrm{MPC}(x,\pi(x,\theta_\pi))]$ is finite valued and continuous on $\Theta_\pi$.
    Condition (vi) is provided by \eqref{eq: uniform convergence of risks}.
    Finally, continuity of $\mathbb{E}_{\mathcal{D}_X}[L_\mathrm{MPC}(x,\pi(x,\theta_\pi))]$ and compactness of $\Theta_\pi$ imply that $\Theta_\pi^{\mathcal{D}_X} \subseteq \Theta_\pi$ is nonempty, thus ensuring (vii).
    This allows to apply \cite[Theorem 5.3]{stochastic_programming} with $C =\Theta_\pi$, providing~\eqref{eq: convergence of parameters}.\par
    3) Next, we show that $\Theta^\mathcal{S}_\pi = \Theta^*_\pi(\mathcal{S})$. To this end, we consecutively show $\Theta^*_\pi(\mathcal{S}) \subseteq \Theta^\mathcal{S}_\pi$ and $ \Theta^*_\pi(\mathcal{S}) \supseteq \Theta^\mathcal{S}_\pi$, implying $\Theta^\mathcal{S}_\pi = \Theta^*_\pi(\mathcal{S})$.
    By definition of $\Theta^*_\pi(\mathcal{S})$, for all $x \in \mathcal{S}$, $\theta^*_\pi \in \Theta^*_\pi(\mathcal{S})$, and all $\theta_\pi \in \Theta_\pi$, we have that
    \begin{equation}
        L_\mathrm{MPC}(x,\pi(x,\theta^*_\pi)) \leq L_\mathrm{MPC}(x,\pi(x,\theta_\pi)).
    \end{equation}
    By monotonicity of the $\mathbb{E}$ operator, for all $x \in \mathcal{S}$, $\theta^*_\pi \in \Theta^*_\pi(\mathcal{S})$, and all $\theta_\pi \in \Theta_\pi$, we get that 
    \begin{equation}
        \mathbb{E}_\mathcal{S}[L_\mathrm{MPC}(x,\pi(x,\theta^*_\pi))] \leq \mathbb{E}_\mathcal{S}[L_\mathrm{MPC}(x,\pi(x,\theta_\pi))].
    \end{equation}
    Hence, for any $\theta^*_\pi \in \Theta^*_\pi(\mathcal{S})$ we have that $\theta^*_\pi$ minimizes $\mathbb{E}_\mathcal{S}[L_\mathrm{MPC}(x,\pi(x,\theta_\pi))$, i.e., $\theta^*_\pi \in \Theta^\mathcal{S}_\pi$. This establishes the assertion $\Theta^*_\pi(\mathcal{S}) \subseteq \Theta^\mathcal{S}_\pi$.
    Next, we show that $\Theta^\mathcal{S}_\pi \subseteq \Theta^*_\pi(\mathcal{S})$, which we do by contradiction. To this end, suppose  $\Theta^\mathcal{S}_\pi \not \subseteq \Theta^*_\pi(\mathcal{S})$ and consider any $\bar{\theta}_\pi \in \Theta^\mathcal{S}_\pi$ such that $\bar{\theta}_\pi \neq \Theta^*_\pi(\mathcal{S})$. Then, for all $\theta^*_\pi \in \Theta^*_\pi(\mathcal{S})$ this implies that
    \begin{equation}
        L_\mathrm{MPC}(x,\pi(x,\theta^*_\pi)) < L_\mathrm{MPC}(x, \pi(x,\bar{\theta}_\pi))
    \end{equation}
    for at least one $x \in \mathcal{S}$. Taking the expectation over $\mathcal{S}$, together with the fact that $p(x)>0$ for all $x \in \mathcal{S}$, we obtain that for all $\theta^*_\pi \in \Theta^*_\pi(\mathcal{S})$ it holds that
    \begin{equation}
       \mathbb{E}_\mathcal{S}[L_\mathrm{MPC}(x,\pi(x,\theta^*_\pi)) < \mathbb{E}_\mathcal{S}[L_\mathrm{MPC}(x,\pi(x,\bar{\theta}_\pi)).
    \end{equation}
    This, however, contradicts the assertion $\bar{\theta}_\pi \in \Theta^\mathcal{S}_\pi$, since by definition of $\Theta^\mathcal{S}_\pi$ for all $\theta_\pi \in \Theta_\pi$ it holds that $\mathbb{E}_\mathcal{S}[L_\mathrm{MPC}(x,\pi(x,\bar{\theta}_\pi))\leq \mathbb{E}_\mathcal{S}[L_\mathrm{MPC}(x,\pi(x,\theta_\pi))$. This implies that $\bar{\theta}_\pi \in \Theta^*_\pi(\mathcal{S})$. Since $\bar{\theta}_\pi \in \Theta^\mathcal{S}_\pi$ was arbitrary, we can conclude $\Theta^\mathcal{S}_\pi \subseteq \Theta^*_\pi(\mathcal{S})$.
    Consequently, we have that $\Theta^\mathcal{S}_\pi \subseteq \Theta^*_\pi(\mathcal{S})$ and $\Theta^*_\pi(\mathcal{S}) \supseteq \Theta^\mathcal{S}_\pi $, implying that $\Theta^*_\pi(\mathcal{S}) = \Theta^\mathcal{S}_\pi$.\par
    4) Combining~\eqref{eq: convergence of parameters} and~\eqref{eq: equality of parameters} yields that
    \begin{equation}\label{eq: convergence empirical ideal parameters}
        \mathbb{D}(\Theta^{\mathcal{D}_X}_\pi,\Theta^*_\pi(\mathcal{S}))\xrightarrow[]{N_s \rightarrow \infty} 0
    \end{equation}
    with probability 1.\par 
    5) Since $\pi(x,\theta_\pi)$ is continuous in $x$ and $\theta_\pi$ on the compact domain $\mathcal{S} \times \Theta_\pi$, we have that $\pi(x,\theta_\pi)$ is uniformly continuous. This implies that for any $\hat{\theta}_\pi \in \Theta^{\mathcal{D}_X}_\pi$ and all $x \in \mathcal{S}$ we get that
    \begin{equation}
        \mathbb{d}(\pi(x,\hat{\theta}_\pi),\Pi^*(x)) \xrightarrow[]{N_s \rightarrow \infty} 0
    \end{equation}
    with probability 1, which concludes the proof.}{}

\end{document}